\documentclass[lettersize,journal]{IEEEtran}
\IEEEoverridecommandlockouts
\let\proof\relax
\let\endproof\relax
\usepackage{mathtools, amsthm, amssymb, amsfonts, xargs, tensor, cite, stmaryrd, mathrsfs, graphicx, xcolor, balance}
\usepackage{pgfplots}
    \pgfplotsset{compat=1.18}
    \usetikzlibrary{positioning, calc, arrows}
\usepackage[hidelinks]{hyperref}

\usepackage[bytype]{Theorems}

\DeclareMathOperator{\diag}{diag}

\DeclareMathOperator{\spanop}{Span}

\DeclareMathOperator{\rangeop}{Range}
\DeclareMathOperator{\uvecop}{uvec}
\DeclareMathOperator{\eigop}{eig}

\begin{document}

\title{Pilot Performance modeling via observer-based inverse reinforcement learning}

\author{Jared Town, Zachary Morrison, and Rushikesh Kamalapurkar \thanks{The authors are with the School of Mechanical and Aerospace Engineering, Oklahoma State University, Stillwater, OK, USA. {\tt\small \{jared.town, zachmor, rushikesh.kamalapurkar\}@okstate.edu}. This research was supported, in part, by the National Science Foundation (NSF) under award numbers 1925147 and 2027999 and the Air Force Office of Scientific Research under award number FA9550-20-1-0127. Any opinions, findings, conclusions, or recommendations detailed in this article are those of the author(s), and do not necessarily reflect the views of the sponsoring agencies.}}
\maketitle
\begin{abstract}
    The focus of this paper is behavior modeling for pilots of unmanned aerial vehicles. The pilot is assumed to make decisions that optimize an unknown cost functional, which is estimated from observed trajectories using a novel inverse reinforcement learning (IRL) framework. The resulting IRL problem often admits multiple solutions. In this paper, a recently developed novel IRL observer is adapted to the pilot modeling problem. The observer is shown to converge to one of the equivalent solutions of the IRL problem. The developed technique is implemented on a quadcopter where the pilot is a linear quadratic supervisory controller that generates velocity commands for the quadcopter to travel to and hover over a desired location. Experimental results demonstrate the robustness of the method and its ability to learn equivalent cost functionals.
\end{abstract}
\section{Introduction}
Given the widespread use of small unmanned aerial systems (sUAS), quadcopters in particular, the need to manage flights efficiently in low altitude settings arises as that airspace is cluttered and turbulent. Cooperative piloting will be necessary for the guidance of these quadcopters to prevent air-to-air and air-to-obstacle collisions. Piloting a small quadcopter in a windy obstacle-laden environment is a difficult task for pilots to do without assistance. We envision a pilot-assist system that recommends paths to the pilots that are personalized to suit their preferences and skill levels. To develop such a system, pilot performance is modeled in terms of a cost functional that is learned by analyzing the pilot's control inputs. While the learned cost functional can be paired with existing optimal control techniques to generate personalized path/trajectory recommendations, optimization is not discussed in this paper. This study exclusively focuses on the cost functional estimation component of the recommendation system.
    
Taking inspiration from \cite{SCC.Xu.Tan.ea2017,SCC.Abbeel.Coates.ea2010}, we hypothesize that the pilot's skill level and preferences can be encoded in a cost functional. We then model the pilot-aircraft system as an optimal control problem where the natural tendencies and skill level \cite{SCC.Phatak.Weinert.ea1976} of the pilot are encoded into a cost functional that the pilot is assumed to optimize. We aim to recover the said cost functional using flight data.
    
Inverse reinforcement learning (IRL) is the process of measuring an ``expert's'' inputs and the resulting behavior over time and obtaining their cost functional. The said expert generates trajectories that are consistent with a given dynamic model, and is assumed to be behaving optimally with respect to some unknown cost functional \cite{SCC.Ng.Russell2000,SCC.Russell1998,SCC.Abbeel.Ng2004,SCC.Abbeel.Ng2005,SCC.Ratliff.Bagnell.ea2006,SCC.Ziebart.Maas.ea2008,SCC.Zhou.Bloem.ea2018,SCC.Levine.Popovic.ea2010,SCC.Neu.Szepesvari2007,SCC.Syed.Schapire2008,SCC.Levine.Popovic.ea2011,SCC.Mombaur.Truong.ea2010, SCC.Lian.Donge.ea2022, SCC.Self.Abudia.ea2022}. A general characteristic of such methods is that they require multiple trajectories and are computationally complex, making them unsuitable for online,  real-time implementation. To address the IRL problem in a real-time, online setting, methods such as \cite{SCC.Self.Coleman.ea2021a,SCC.Self2020, SCC.Rhinehart.Kitani2018, SCC.Herman.Fischer.ea2015, SCC.Arora.Doshi.2017a} have been developed. These methods are typically model-based and use a single continuous trajectory to learn the cost functional of an expert. A notable result is obtained in \cite{SCC.Lian.Xue.ea2021a} where an online and model-free approach that utilizes a neural network to solve the IRL problem in the presence of adversarial attacks is developed. However, this method only identifies the state penalty matrix and is unable to identify the control penalty matrix. 
    
This paper is focused on the development of an IRL formulation of the pilot modeling problem and an adaptation of the of the regularized history stack observer (RHSO) developed in \cite{town2022nonuniqueness} to solve the resulting IRL problem. It is shown in  \cite{SCC.Jean.Maslovskaya2018} that IRL problems that have a product structure have multiple linearly independent solutions. Since the linearized model of a quadcopter decouples lateral and longitudinal dynamics, it has a product structure. As a result, implementation of IRL to estimate cost functionals of quadcopter pilots requires methods such as \cite{town2022nonuniqueness} that are suited for IRL problems with multiple linearly independent solutions.
    
The method developed in \cite{town2022nonuniqueness} is an online IRL method that is capable of identifying the true cost functional of the pilot, up to a scaling factor, if the IRL problem has a unique (up to a scaling factor) solution, and an equivalent solution (that is, a cost functional that results in the same feedback matrix as the expert), if the IRL problem admits multiple linearly independent solutions. The key contribution of this paper is a formulation of the pilot modeling problem in the framework of IRL where the pilot's control inputs are velocity commands that are executed by an onboard autopilot. The formulation allows for the use of the IRL method in \cite{town2022nonuniqueness}, with minimal modification, to estimate a cost functional that models a pilot's performance.     
	    
The paper is structured as follows: Section \ref{Section: Pilot Model IRL} contains the assumed pilot model and linearized quadcopter model. Section \ref{Section: IRL} contains IRL method. Section \ref{Section: Experiments} contains experimental setup, results and discussion. Section \ref{Section: Conclusion} concludes the paper.

\section{Formulation of the pilot performance modeling problem in an IRL framework}\label{Section: Pilot Model IRL}
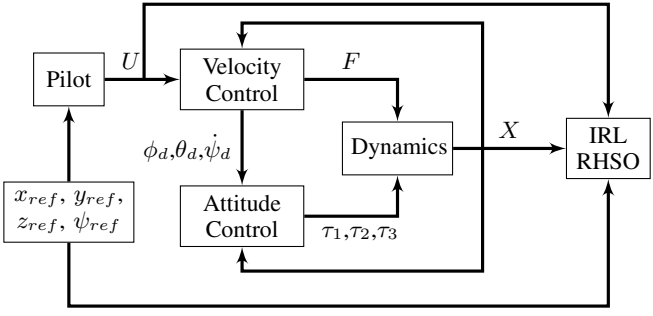
\begin{figure}
    \begin{center}
        \tikzset{
            font=\small,
            block/.style={draw,align=center,minimum height=2em,text width=3em},
            arr/.style={>=latex',very thick},
        }
        \begin{tikzpicture}
            \node [block,text width=2em, minimum height=2em] (pilot) at (-7,0) {Pilot};
            %\node [below=1cm of pilot] (ref) {\Large \textopenbullet};
            \node [block, text width=4.25em, below=1cm of pilot] (ref) {$x_{ref}$, $y_{ref}$, $z_{ref}$, $\psi_{ref}$};
            \node [block,text width=4em, right=1cm of pilot] (vctl) {Velocity Control};
            \node [block,text width=4em, below=1cm of vctl] (actl) {Attitude Control};
            \node [block,text width=3.5em, below right=0.15cm and 0.5cm of vctl] (dyn) {Dynamics};
            \node [block,text width=2.5em, right=1.5cm of dyn] (irl) {IRL RHSO};
            \draw [->,arr] (pilot) -- (vctl) node [above,pos=0.35, align=center] {$U$};%
            \draw [->,arr] (vctl) -- (actl) node [left,pos=0.5] {$\phi_d$,$\theta_d$,$\dot{\psi}_d$};
            \draw [->,arr] (vctl) -| (dyn) node [above,pos=0.25] {$F$};
            \draw [->,arr] (actl) -| (dyn) node [below,pos=0.3] {$\tau_1$,$\tau_2$,$\tau_3$};
            \draw [->,arr] (dyn) -- (irl) node [above,pos=0.5] {$X$};
            \draw [->,arr] (-6,0) -- ++(0,1) -| (irl) node [below,pos=0.5] {};
            \draw [->,arr] (-1.5,-0.9) -- ++(0,1.65) -| (vctl) node [below,pos=0.5] {};
            \draw [->,arr] (-1.5,-0.9) -- ++(0,-1.65) -| (actl) node [below,pos=0.5] {};
            \draw [->,arr] (ref) |- ++(0,-1.25) -| (irl) node [right,pos=0, align=center] {};
            \draw [->,arr] (ref)  -- (pilot) node [right,pos=0, align=center] {};
        \end{tikzpicture}
    \end{center}
    \caption{Pilot and Quadcopter Combined Model}
    \label{fig:Model}
\end{figure}

\subsection{Problem Statement}
This study concerns a quadcopter sUAS with an onboard autopilot being flown by a human pilot via desired velocity commands. That is, from the perspective of the human pilot, the control input is the desired linear velocities of the quadcopter and the desired yaw rate. The human pilot is asked to regulate the aircraft to the origin, starting from a non-zero initial condition. The objective is to find a best-fit cost functional such that a controller that optimizes the cost functional results in trajectories that are similar to those observed under human control. 

In this proof-of-concept study, we assume that the human pilot can observe the full state of the quadcopter and the experimental study utilizes supervisory LQR controllers as surrogates in lieu of human pilots. The control commands sent to the aircraft by the LQR surrogates, along with the full state of the quadcopter, are used to learn the surrogate pilot's cost functional using an observer-based inverse reinforcement learning (IRL) algorithm. Since the IRL problem formulated in Section \ref{Section: Pilot Model IRL} admits multiple solutions, we aim to recover an equivalent cost functional, per Definition \ref{Definition: Equivalent Solution Pilot}.

\subsection{Pilot Model}
The pilot-controlled system is assumed to be a linear time-invariant system of the form
\begin{equation}\label{Equation: Pilot System}
    \dot{X}(t) = AX+BU,
\end{equation}
where the state is $X\in\mathbb{R}^{12}$ and the control input is $U\in\mathbb{R}^{4}$. The system matrices are given as $A \in \mathbb{R}^{12\times 12}$ and $B \in \mathbb{R}^{12\times 4}$. The pilot is assumed to employ an optimal controller that optimizes the cost functional
\begin{equation}\label{Equation: Pilot Cost Fun}
    J(X_0,U({\cdot}))=
    \int_{0}^{\infty} \left(X(t)^{\top}QX(t) + U(t)^{\top}RU(t)\right)dt, 
\end{equation}
where $X(\cdot)$ is the system trajectory under the control signal $U(\cdot)$, starting from the initial condition $X_0$, $Q \in\mathbb{R}^{12\times 12}$ is an unknown positive semi-definite matrix, and $R\in\mathbb{R}^{4\times 4}$ is an unknown positive definite matrix.

\begin{assumption}\label{Assumption: Stab_Det}
    The pair $(A,B)$ is stabilizable and $(A,\sqrt{Q})$ is detectable. 
\end{assumption}
Stabilizability of $(A,B)$ and detectability of $(A,\sqrt{Q})$ is needed for the optimal controller to exist. Linearized models of quadrotors, including the one used in the experiments presented in Section \ref{Section: Experiments}, are stabilizable. In the experiments, the pilot is assumed to penalize translational position errors and heading errors, and the resulting pair $(A,\sqrt{Q})$ is shown to satisfy the detectability condition.

The algebraic Riccati equation (ARE),
\begin{equation}\label{Equation: ARE}
    A^{\top}S+SA-SBR^{-1}B^{\top}S+Q=0,
\end{equation}
with respect to the optimal control problem described by \eqref{Equation: Pilot System} and \eqref{Equation: Pilot Cost Fun} can then be solved, which yields the optimal policy of the pilot, given by $u = K_{EP}x$, where $K_{EP} = R^{-1}B^{\top}S$.

The surrogate pilot's policy is recovered by estimating, online, and in real-time, the unknown matrices $Q$ and $R$, using the known system matrices, $A$ and $B$, given measurements of $X$ and $U$.

\subsection{Quadcopter Model}
To implement the developed model-based IRL method, a linearized quadcopter model, with velocity commands as the input, and the actual position, velocity, orientation, and angular velocity as the output needs to be developed. Such a model depends on the autopilot being used to stabilize the aircraft, and as such, knowledge of the autopilot algorithm is required to complete the model. Note that identification of the autopilot is not the focus of this study, we assume that the autopilot is able to track the commanded inputs, and aim to model the cost functional of a surrogate LQR pilot that generates velocity commands that are then implemented by the autopilot.

The model used in this study closely follows the development in \cite{SCC.Islam.Okasha.ea2017, SCC.Bouabdallah.Siegwart2007, SCC.Bouabdallah.Noth.ea2004}
The state variables of the model are 
\begin{equation*}
    X \coloneqq \left[x, y, z, \dot{x}, \dot{y}, \dot{z}, \phi, \theta, \psi, \dot{\phi}, \dot{\theta}, \dot{\psi}\right]^{\top},
\end{equation*}
where $x$, $y$, and $z$, are the translational positions, $\dot{x}$, $\dot{y}$, and $\dot{z}$, are the translational velocities. Also, $\phi$, $\theta$, and $\psi$, are the roll pitch and yaw angular positions and $\dot{\phi}$, $\dot{\theta}$, and $\dot{\psi}$ are the respective angular velocities. The control input is given by
\begin{equation*}
    U \coloneqq \left[\dot{x}_d, \dot{y}_d, \dot{z}_d, \dot{\psi}_d\right]^{\top},
\end{equation*}
where $\dot{x}_d$, $\dot{y}_d$, and  $\dot{z}_d$, are the desired translational velocities and $\dot{\psi}_d$ as the desired heading angular velocity.
The translational dynamics of a quadcopter are described in the North, East, Down (NED) coordinate frame  by \cite{SCC.Islam.Okasha.ea2017}
\begin{equation}\label{Equation: Quad Trans}
    m \begin{bmatrix}
        \ddot{x} \\ \ddot{y} \\ \ddot{z}
    \end{bmatrix} = 
    \begin{bmatrix}
        0 \\ 0 \\ mg
    \end{bmatrix} +
    R \begin{bmatrix}
      0 \\ 0 \\ -F  
    \end{bmatrix} -
    k_t \begin{bmatrix}
        \dot{x} \\ \dot{y} \\ \dot{z}
    \end{bmatrix},
\end{equation}
where $k_t$ is the aerodynamic drag, $m$ is the mass, and $g$, is the acceleration due to gravity, and $R$ is the rotational matrix where small angle approximations result in 
\begin{equation}
    R=\begin{bmatrix}
        1 & \phi\theta-\psi & \theta+\phi\psi \\
        \psi & \phi\theta\psi+1 & \theta\psi - \phi \\
        -\theta & \phi & 1
    \end{bmatrix}.
\end{equation}
The thrust, $F$, applied by the autopilot is a proportional controller
\begin{equation}\label{Equation: Quad U1}
    F = mg + mk_{p_{13}}(\dot{z}-\dot{z}_d).
\end{equation}
The rotational motion of the quadcopter is described by \cite{SCC.Bouabdallah.Noth.ea2004, SCC.Bouabdallah.Siegwart2007}
\begin{equation}\label{Equation: Quad Rot}
    \begin{aligned}
        \ddot{\phi}I_{xx} &= \dot{\theta}\dot{\psi}(I_{yy}-I_{zz}) + l\tau_1,\\
        \ddot{\theta}I_{yy} &= \dot{\phi}\dot{\psi}(I_{zz}-I_{xx}) + l\tau_2,\\
        \ddot{\psi}I_{zz} &= \dot{\theta}\dot{\phi}(I_{xx}-I_{yy}) + \tau_3,
    \end{aligned}
\end{equation}
with $I_{xx}$, $I_{yy}$, and $I_{zz}$ being moment of inertia and $\tau_1$, $\tau_2$, and $\tau_3$ being torques designed as
\begin{equation}\label{Equation: QUAD U234}
    \begin{aligned}
        \tau_1 &= k_{p_{21}}(\phi_d-\phi) - k_{d_{1}}\dot{\phi},\\
        \tau_2 &= k_{p_{22}}(\theta_d-\theta) - k_{d_{2}}\dot{\theta},\\
        \tau_3 &= k_{d_{3}}(\dot{\psi}_d-\dot{\psi}).
    \end{aligned}
\end{equation}
The desired angles $\phi_d$ and $\theta_d$, commanded by the autopilot, are given by
\begin{equation}\label{Equation: Autopilot}
    {\thickmuskip=0mu \thinmuskip=0mu \medmuskip=0mu \begin{bmatrix}
        \theta_d \\ \phi_d
    \end{bmatrix}=
    \begin{bmatrix}
        \arctan\left(\frac{ k_{p_{12}}(\dot{y}_d-\dot{y})\sin\psi+k_{p_{11}}(\dot{x}_d-\dot{x})\cos\psi }{g+k_{p_{13 }}(\dot{z}_d-\dot{z})}\right)\\
        \arctan\left(cos\theta_d\frac{ k_{p_{11}}(\dot{x}_d-\dot{x})\sin\psi-k_{p_{12}}(\dot{y}_d-\dot{y})\cos\psi }{g+k_{p_{13 }}(\dot{z}_d-\dot{z})}\right)
    \end{bmatrix},}
\end{equation}
where $k_{p_{11}}$, $k_{p_{12}}$, $k_{p_{13}}$, $k_{p_{21}}$, $k_{p_{22}}$, $k_{d_{1}}$, $k_{d_{2}}$, $k_{d_{3}}$ are control gains of the autopilot. The desired angles are simplified with small angle approximations and a linear approximation for the inverse tangent function \cite{SCC.Rajan.Wang.ea2006} to yield
\begin{equation}\label{Equation: QUAD THETAPHI}
    \begin{aligned}
        \theta_d &= \frac{\pi}{4}\left(\frac{ k_{p_{12}}(\dot{y}_d-\dot{y})\psi+k_{p_{11}}(\dot{x}_d-\dot{x})}{g+k_{p_{13 }}(\dot{z}_d-\dot{z})}\right),\\
        \phi_d &= \frac{\pi}{4}\left(\frac{ k_{p_{11}}(\dot{x}_d-\dot{x})\psi-k_{p_{12}}(\dot{y}_d-\dot{y})}{g+k_{p_{13 }}(\dot{z}_d-\dot{z})}\right).
    \end{aligned}
\end{equation}
Linearizing \eqref{Equation: Quad Trans} and \eqref{Equation: Quad Rot} about the origin, while using \eqref{Equation: Quad U1}, \eqref{Equation: QUAD U234}, and \eqref{Equation: QUAD THETAPHI}, yields the linear system
\begin{equation}\label{Equation: Quad Linear_Sys}
    \begin{aligned}
        \ddot{x} &= -g\theta-\frac{k_t}{m}\dot{x},\\
        \ddot{y} &= g\phi - \frac{k_t}{m}\dot{y},\\
        \ddot{z} &= k_{p_{13}}(\dot{z}_d-\dot{z}) - \frac{k_t}{m}\dot{z},\\
        \ddot{\phi} &= \frac{b_1\pi k_{p_{21}}k_{p_{12}}(\dot{y}-\dot{y}_d)}{4g} -b_1k_{d_{1}}\dot{\phi}-b_1k_{p_{21}}\phi,\\
        \ddot{\theta} &= \frac{b_2\pi k_{p_{22}}k_{p_{11}}(\dot{x}_d-\dot{x})}{4g} -b_2k_{d_{2}}\dot{\theta}-b_2k_{p_{22}}\theta,\\
        \ddot{\psi}, &= b_3k_{d_{3}}(\dot{\psi}_d-\dot{\psi}),
    \end{aligned}
\end{equation}
where $b_1 = \frac{l}{I_{xx}}$, $b_2 = \frac{l}{I_{yy}}$, and $b_3 = \frac{1}{I_{zz}}$, and $l$ is the length of the quadcopter arm. 

As shown in Figure \ref{fig:Model}, given measurements of the state variables, i.e., translational positions $[x,y,z]$, translational velocities $[\dot{x},\dot{y},\dot{z}]$, angular positions $[\phi,\theta,\psi]$, and angular velocities $[\dot{\phi}, \dot{\theta}, \dot{\psi}]$, and the control variables, i.e., the desired velocities $[\dot{x}_d,\dot{y}_d,\dot{z}_d]$ and yaw rate $[\dot{\psi}_d]$ commanded by the LQR surrogate pilot, we aim to find an equivalent solution ($\hat{Q}$, $\hat{S}$, $\hat{R}$) of the IRL problem according to the following definition.
\begin{definition}\label{Definition: FI}(\hspace{-0.01em}\cite{town2022nonuniqueness})\label{Definition: Equivalent Solution Pilot}
    Given $\varpi \geq 0$, a solution ($\hat{Q}$, $\hat{S}$, $\hat{R}$) to the IRL problem is called an $\varpi-$equivalent solution of the IRL problem if 
    \begin{equation*}
        \left\Vert A^{\top}\hat S+\hat S A- \hat SB\hat R^{-1}B^{\top}\hat S+\hat Q\right\Vert\leq \varpi,
    \end{equation*}
    and optimization of the performance index $J$, with $Q=\hat{Q}$ and $R=\hat{R}$, results in a feedback matrix, $\hat{K}_{p} \coloneqq \hat{R}^{-1}B^{\top} \hat{S}$, that satisfies
    \begin{equation*}
        \left\Vert \hat{K}_{p} - K_{EP}\right\Vert \leq \varpi.
    \end{equation*}
\end{definition}
\section{Inverse Reinforcement Learning}\label{Section: IRL}
This section illustrates the IRL algorithm used to estimate a cost functional that is equivalent to the pilot's cost functional. The algorithm is similar to \cite{town2022nonuniqueness}, with minor modifications to account for availability of full state measurements.

\subsection{The Regularized History Stack Observer}
The following development is a special case of the RHSO developed in \cite{town2022nonuniqueness}, where the system state is assumed to be measurable. The state estimates generated by the onboard Kalman filter are used in the experiment to obtain an equivalent solution to the IRL problem per Definition \ref{Definition: Equivalent Solution Pilot}, The RHSO is constructed as follows.
 
Using the Popov–Belevitch–Hautus (PBH) test in Theorem 14.3 of \cite{ SCC.Hespanha2009}, it can be shown that the pilot model developed in Section \ref{Section: Pilot Model IRL} satisfies the stabilizability condition in Assumption \ref{Assumption: Stab_Det}. If $Q$ is assumed to meet the detectability condition in Assumption \ref{Assumption: Stab_Det} and if the state and control trajectories, $X(\cdot)$ and $U(\cdot)$ respectively, of the quadcopter, are optimal with respect to the cost functional in \eqref{Equation: Pilot Cost Fun}, then there exists a matrix $S$ such that $Q$, $R$, $A$, $B$, and $S$ satisfy the Hamilton-Jacobi-Bellman (HJB) equation
\begin{equation}\label{Equation: IRL HJB}
    X^{\top}(t)\left(A^{\top}S+SA-SBR^{-1}B^{\top}S+Q\right)X(t)=0,
\end{equation}
and the optimal control equation
\begin{equation}\label{Equation: Optimal Ctrl}
    U(t) = -R^{-1}B^{\top}SX(t),
\end{equation}
for all $t \in\mathbb{R}_{\ge 0}$. The linear system in \eqref{Equation: Quad Linear_Sys} is comprised of two decoupled systems. If the penalty matrices $Q$ and $R$ are also decoupled, for example, diagonal, then the corresponding IRL problem admits multiple solutions \cite{SCC.Jean.Maslovskaya2018}. 

From Definition \ref{Definition: Equivalent Solution Pilot}, it can be concluded that if $\hat{Q}$, $\hat{R}$, and $\hat{S}$ are parts of a $\varpi-$equivalent solution to the IRL problem with $\varpi = 0$, and if $\hat{R}$ is invertible, then for all $t \in\mathbb{R}_{\ge 0}$,
\begin{equation}\label{Equation: IRL HJB Hat}
    X^{\top}(t)\left(A^{\top}\hat S+\hat SA-\hat SB\hat R^{-1}B^{\top}\hat S+\hat Q\right)X(t)=0
\end{equation}
and
\begin{equation}\label{Equation: Optimal Ctrl Hat}
    U(t) = -\hat R^{-1}B^{\top}\hat SX(t).
\end{equation}

Given measurements of the the state, $X$, and control signal, $U$, and estimates $\hat Q$, $\hat R$, and $\hat S$ of $Q$, $R$, and $S$, respectively, \eqref{Equation: IRL HJB Hat} and \eqref{Equation: Optimal Ctrl Hat} can be used to develop an equivalence metric that evaluates to zero if the estimates constitute an equivalent solution. Since scaling of a cost functional results in another equivalent cost functional, equivalent cost functionals can only be identified up to a scaling factor. To fix the scale, the (1,1) element of $\hat{R}$, denoted by $r_1$, is selected to be equal to one. In particular, the RHSO generates an equivalent solution  $(\hat Q,\hat R,\hat S)$  using
\begin{equation}\label{Equation: IRL Update Law}
    \dot{\hat{W}}= 
     (\Sigma^{\top}\Sigma + \epsilon I)^{-1}\Sigma^{\top}
    \left(\Sigma_u - \Sigma\hat{W}\right).
    \end{equation}
In \eqref{Equation: IRL Update Law}, $\hat{W} = \begin{bmatrix}
    \hat{W}_S^\top,& \hat{W}_Q^\top,& (\hat{W}_R^-)^\top
\end{bmatrix}^\top$, where $\hat{W}_S\in\mathbb{R}^{78}$, $\hat{W}_Q\in\mathbb{R}^{78}$, and $\hat{W}_R\in\mathbb{R}^{10}$ are weights that satisfy $(\hat{W}_{S})^{\top}\sigma_S(X) = X^{\top}\hat{S}X$, $(\hat{W}_{Q})^{\top}\sigma_Q(X) = X^{\top}\hat{Q}X$, $(\hat{W}_{R})^{\top}\sigma_{R1}(U) = U^{\top}\hat{R}U$ and $\sigma_{R2}(U)\hat{W}_{R} = \hat{R}U$, respectively, and the vector $\hat{W}_R^-$ is a copy of $\hat{W}_R$ with the first element, $r_1$, removed. The basis functions are given by
\begin{align*}
    &\!\begin{multlined}
        \sigma_S(X)=\sigma_Q(X):=[X_1^2,2X_1X_2,2X_1X_3,\ldots,2X_1X_{12},\\
        X_2^2,2X_2X_3,2X_2X_4,\ldots,X_{11}^2,\ldots,2X_{11}X_{12},X_{12}^2]^T,
    \end{multlined}\\
    &\!\begin{multlined}
        \sigma_{R1}(U):=[U_1^2,2U_1U_2,2U_1U_3,2U_1U_4,U_2^2,2U_2U_3,\\
        2U_2U_4,U_{3}^2,2U_{3}U_4,U_4^2]^T,
    \end{multlined}
\end{align*}
and
\begin{equation}
    \sigma_{R 2}(U)=
    \begin{bmatrix}
        U^{T} & 0_{1 \times 3} & 0_{1 \times 2} & 0 \\
        U_{1}e_{2,4} & \left(U^{T}\right)^{(-1)} & 0_{1 \times 2} & 0 \\
        U_{1}e_{3,4} & U_{2}e_{2,3} & \left(U^{T}\right)^{(-2)} & 0 \\
        U_{1}e_{4,4} & U_{2}e_{3,3} & U_{3}e_{2,2} & U_{4}
    \end{bmatrix},\label{eq:u_basis}
\end{equation} 
where $U^{(-j)}$ denotes the vector $U$ with the first $j$ elements removed, and $e_{i,j}$ denotes a row vector of size $j$, with a one in the $i-$th position and zeros everywhere else. The matrices $\Sigma \in \mathbb{R}^{165N}\times \mathbb{R}^{165}$ and $\Sigma_u \in \mathbb{R}^{165N}$, referred to collectively as \emph{the history stack}, are constructed as
\begin{align*}
    \Sigma \coloneqq
    \begin{bmatrix} \sigma_\delta\left(X(t_1),U(t_1)\right)\\ \sigma_{\Delta_u}\left(X(t_1),U(t_1)\right)\\
    \vdots\\
    \sigma_\delta\left(X(t_N),U(t_N)\right)\\ \sigma_{\Delta_u}\left(X(t_N),U(t_N)\right)
    \end{bmatrix},
    &&
    \Sigma_u \coloneqq
    \begin{bmatrix}
    -U_1^2(t_1)r_1\\
    -2U_1(t_1)r_1\\
    0_{m-1\times 1}
    \\
    \vdots\\
    -U_1^2(t_N)r_1\\
    -2U_1(t_N)r_1\\
    0_{m-1\times 1}
    \end{bmatrix},
\end{align*}
where $N$ is the number of time instances selected for storage and the functions $\sigma_{\delta}$ and $\sigma_{\Delta_u}$ are given by
\begin{multline}\label{Equation: IRL Sig_Delta}
    \sigma_{\delta}\left(X,U\right) =\\
    \begin{bmatrix} (AX+BU)^{\top}(\nabla_X\sigma_S(X))^{\top} & \sigma_Q(X)^{\top} & \sigma_{R1}^-(U)^{\top}
    \end{bmatrix}
\end{multline}
and
\begin{multline}\label{Equation: IRL Sig_Delta_u}
    \sigma_{\Delta_u}\left(X,U\right)=\\
    \begin{bmatrix}
     B^{\top}(\nabla_x\sigma_S(X))^{\top} & 0_{4\times 78} & 2\sigma_{R2}^-(U)
    \end{bmatrix},
\end{multline}
where $\sigma_{R2}^{-}$ is a copy of $\sigma_{R2}$ with the first column removed and $\sigma_{R1}^-$ is a copy of $\sigma_{R1}$ with the first element removed. 

Theorem \ref{Theorem: Delta Convergence IRL} below, which guarantees convergence of \eqref{Equation: IRL Update Law} to an equivalent solution, relies on the formulation of an error metric $\Delta(t) \coloneqq \Sigma_u-\Sigma\hat{W}(t)$ and its time derivative
\begin{equation}\label{Equation: Delta_dot}
    \dot{\Delta} = -\Sigma(\Sigma^T\Sigma+\epsilon I)^{-1}\Sigma^T\Delta,
\end{equation}
along with the following data informativity condition adopted from \cite{town2022nonuniqueness}.
\begin{definition}\label{Definition: FI IRL}
    The signal $(X,U)$ is called finitely informative (FI) if there exists a time instance $T >0 $ such that for some $\{t_1, t_2, \ldots , t_N\} \subset [0,T] $,
    \[
        \spanop\left\{X(t_i)\right\}_{i=1}^N = \mathbb{R}^n,
    \]
    \[
        \spanop\left\{X(t_i) X(t_i)^{\top}\right\}_{i=1}^N = \{\mathbb{Z}\in \mathbb{R}^{n\times n}| \mathbb{Z}=\mathbb{Z}^{\top}\}, \quad\text{and}
    \]
    \begin{equation}
        \Sigma_u \in \rangeop(\Sigma).\label{Equation: Sigma_u Condition}
    \end{equation}
    In addition, for a given $\epsilon > 0$, if $\min\{\eigop(X X^\top)\} > \epsilon$ and $\min\{\eigop(Z Z^\top)\} > \epsilon$, where $X \coloneqq [\hat{x}(t_1),\ldots,\hat{x}(t_N)]$, $Z \coloneqq [\uvecop(\hat x(t_1) \hat{x}^{\top}(t_1)),\ldots,\uvecop(\hat x(t_N) \hat{x}^{\top}(t_N))] \in\mathbb{R}^{\frac{n(n+1)}{2}\times N}$, and $\uvecop(\hat x(t_i) \hat{x}^{\top}(t_i)) \in \mathbb{R}^{\frac{n(n+1)}{2}}$ denotes vectorization of the upper triangular elements of the symmetric matrix $\hat x(t_i) \hat{x}^{\top}(t_i) \in\mathbb{R}^{n\times n}$, then $(\hat{x},u)$ is called \emph{$\epsilon-$finitely informative (FI)}.
\end{definition}
To implement the developed observer, a method to select the time instances $t_1,\ldots,t_N$ is needed. The convergence result summarized in Theorem \ref{Theorem: Delta Convergence IRL} relies on the existence of a time instance $\underline{T} \geq 0$ such that the three conditions in Definition \ref{Definition: FI IRL} are met. As such, any data selection algorithm that ensures the satisfaction of those three conditions can be used to implement the developed observer. In this paper, a data selection method that minimizes the condition number of $\Sigma^{\top}\Sigma + \epsilon I$ is utilized. Minimization of the condition number of $\Sigma^{\top}\Sigma + \epsilon I$ improves the accuracy of matrix inversion in the update law \eqref{Equation: IRL Update Law} and improves the convergence rate of $\Delta$ in \eqref{Equation: Delta_dot}.

In particular, the matrices $\Sigma_u$ and $\Sigma$, contained in the history stack, are recorded at specific time instances according to the following procedure. Both matrices are initialized as zero matrices. Data are then added to the matrices at a user-selected sampling interval until they are filled. Then, a condition number minimization algorithm, similar to \cite{SCC.Kamalapurkar2018}, is used to replace old data with new data, where replacement is carried out only if the post-replacement condition number of $\Sigma^{\top}\Sigma + \epsilon I$ is lower than its pre-replacement condition number. Due to the replacement procedure, the time instances $t_i$ corresponding to data stored in the history stack are piecewise constant functions of time.
        
\begin{theorem}\label{Theorem: Delta Convergence IRL}
    If there exists $\underline{T} \geq 0$ such that $\Sigma_u (t) \in \rangeop(\Sigma(t))$ for all $t\geq \underline{T}$, $\epsilon \geq 0$ is selected to ensure invertibility of $\Sigma^{\top}(t)\Sigma(t)+\epsilon I$ for all $t\geq \underline{T}$, $\min\{\eigop(X(t) X(t)^\top)\} > \epsilon$ and $\min\{\eigop(Z(t) Z(t)^\top)\} > \epsilon$, for some $\epsilon > 0$ and all $t\geq \underline{T}$, with $X$ and $Z$ as introduced in Definition \ref{Definition: FI}, and if there exist a constant $0\leq \underline{R} < \infty$ such that the matrix $\hat{R}(t)$, extracted from $\hat{W}(t)$, is invertible with $\|\hat R^{-1}(t)\|\leq \underline{R}$ for all $t\geq \underline{T}$, then the matrices $\hat{Q}$, $\hat{S}$, and $\hat{R}$, extracted from $\hat{W}$, converge to a $0-$equivalent solution of the IRL problem.
\end{theorem}
\begin{proof}
    The proof, included here for completeness, is a slight modification of the proof of Theorem 7 and Corollary 10 of \cite{town2022nonuniqueness}. Applying Theorem 7 in \cite{town2022nonuniqueness}, with $K_4 = I$, it can be concluded that along the solutions of \eqref{Equation: IRL Update Law}, $\lim_{t\to\infty} \Delta(t) = 0$. Note that the error metric $\Delta$ can be expressed using the basis functions in \eqref{Equation: IRL Sig_Delta}, \eqref{Equation: IRL Sig_Delta_u}, and \eqref{eq:u_basis} as 
    \[
        \Delta = \begin{bmatrix} \sigma_\delta^\prime\left(X(t_1(t)),U(t_1(t))\right)\\ \sigma_{\Delta_u}^\prime\left(X(t_1(t)),U(t_1(t))\right)\\
        \vdots\\
        \sigma_\delta^\prime\left(X(t_N(t)),U(t_N(t))\right)\\ \sigma_{\Delta_u}^\prime\left(X(t_N(t)),U(t_N(t))\right)
        \end{bmatrix}\hat{W}^\prime,
    \]
    where $\hat{W}^\prime\coloneqq\begin{bmatrix} \hat{W}_S^{\top} & \hat{W}_Q^{\top} & \hat{W}_R^{\top} \end{bmatrix}^{\top}$, $\sigma_\delta^\prime(X,U) = [(AX+BU)^{\top}(\nabla_X\sigma_S^{\top}(X) \ \ \ \sigma_Q^{\top}(X) \ \ \ \sigma_{R1}^{\top}(U)]$ and $ \sigma_{\Delta_u}^\prime(X,U) = [ B^{\top}(\nabla_x\sigma_S^{\top}(X) \ \ \  0_{4\times 78} \ \ \  2\sigma_{R2}(U)]$. Using the fact that $\sigma_{\Delta_u}^\prime\left(X(t_i(t)),U(t_i(t))\right)\hat{W}^\prime(t) = \hat R(t)\tilde{K}_P(t)X(t_i(t))$, where $\tilde{K}_P(t) \coloneqq \hat{K}_{p}(t) - K_{EP} $, it can be concluded that
    \begin{multline} \label{Equation:Kp Error Bound}
        \left\Vert \tilde{K}_P(t)X(t_i(t))\right\Vert \le \\
        \left\Vert \hat{R}^{-1}(t) \sigma_{\Delta_u}^\prime\left(X(t_i(t)),U(t_i(t))\right) \hat{W}^\prime(t)\right\Vert.
    \end{multline}
    
    Given any $\varpi > 0$, if $\min\{\eigop(X(t) X(t)^\top)\} > \epsilon$ then there exists $c > 0$, independent of $t$, such that $\left\Vert \tilde{K}_P(t)X(t_i(t))\right \Vert \le \frac{\varpi}{c}$, for all $i=1,\ldots,N $, implies $\left \Vert \tilde{K}_P(t) \right \Vert \le \varpi$. Select $\overline{T}$ large enough such that for all $t\geq \overline{T}$, $\left\Vert\Delta(t)\right\Vert \leq \frac{\varpi }{2c\underline{R}}$. Then, for all $i=1,\ldots,N $, $\left\Vert\sigma_{\Delta_u}^\prime\left(X(t_i(t)),U(t_i(t))\right)\hat{W}^\prime(t)\right\Vert \leq \frac{\varpi}{c\underline{R}}$, which implies $\left\Vert \hat{R}^{-1}(t) \sigma_{\Delta_u}^\prime\left(X(t_i(t)),U(t_i(t))\right) \hat{W}^\prime(t)\right\Vert \leq \frac{\varpi }{2c} $. From \eqref{Equation:Kp Error Bound}, it follows that $\left\Vert \tilde{K}_P(t)X(t_i(t))\right \Vert \le \frac{\varpi}{c}$, and as a result, $\left \Vert \tilde{K}_P(t) \right \Vert \le \varpi$. Since $\varpi$ was arbitrary, $\lim_{t\to\infty} \hat{K}_{p}(t) = K_{EP}$.

    The function $\sigma_{\delta}^\prime$ can be expressed as
    \begin{multline}
        \sigma_{\delta}^\prime\left(X(t_i(t)), U(t_i(t))\right)\hat{W}^\prime(t) = \\
        X^{\top}(t_i(t))\hat M X(t_i(t)) + g\left(\hat{K}_{p}(t), K_{EP}\right),
    \end{multline}
    where the function $g$ satisfies\footnote{For a positive function $g$, $f=O(g)$ if there exists a constant $M$ such that $\left \Vert f(x) \right\Vert \le Mg(x), \forall x$.} $g=O\left(\left\Vert\tilde{K}_{P}(t)\right\Vert\right)$ and $\hat{M}(t)=\left(A^{\top}\hat S(t) + \hat S(t)A-\hat S(t)B\hat R^{-1}(t)B^{\top}\hat S(t)+\hat Q(t)\right)$. Using the triangle inequality, 
    \begin{multline} \label{Equation: M hat bound}
        \left\vert X^{\top}(t_i(t))\hat M(t) X(t_i(t)) \right\vert \leq \\
        \left\vert \sigma_{\delta}^\prime\left(X(t_i(t)),U(t_i(t))\right)\hat{W}^\prime\right\vert + \left\vert g\left(\hat{K}_{p}(t), K_{EP}\right) \right\Vert
    \end{multline}
    Since $\lim_{t\to\infty} \hat{K}_{p}(t) = K_{EP}$, $\lim_{t\to\infty} \Delta(t) = 0$, and $\left\vert\sigma_{\delta}^\prime\left(X(t_i(t)),U(t_i(t))\right)\hat{W}^\prime\right\vert \leq \left\Vert\Delta(t)\right\Vert$, given any $\varepsilon > 0$, the bound in \eqref{Equation: M hat bound} implies that there exists $\overline{T}\geq 0$ such that for all $t\geq \overline{T}$ and for all $i=1,\cdots,N$, $ \left\vert X^{\top}(t_i(t))\hat M(t) X(t_i(t)) \right\vert \leq \varepsilon $. 
     
    Similar to the proof of Corollary 10 in \cite{town2022nonuniqueness}, if $\min\{\eigop(Z(t) Z(t)^\top)\} > \epsilon,\forall t\geq\underline{T}$, then given $\varpi > 0$, one can construct a $\varepsilon > 0$ such that $ \left\vert X^{\top}(t_i(t))\hat M(t) X(t_i(t)) \right\vert \leq \varepsilon $, for all $i=1,\cdots,N$, implies that $\left\Vert \hat{M}(t) \right\Vert \le \varpi$. Therefore, $\lim_{t\to\infty} \left\Vert \hat{M}(t) \right\Vert = 0$, which completes the proof of the theorem.
\end{proof}        

\section{Experiments}\label{Section: Experiments}
Experimental results obtained using the developed RHSO, implemented on a quadcopter, are presented in this section. The pilot is assumed to be a surrogate LQR controller that mimics velocity commands sent by a remote controller to a quadcopter. The velocity commands are treated as desired velocities that are executed by the onboard autopilot. The pilot behavior modeling problem is reformulated as an IRL problem and the ability of the developed IRL method to learn an equivalent solution of the IRL problem using measurements of the quadcopter state and the velocity commands sent by the surrogate LQR controller is demonstrated.

\subsection{Hardware}
A custom-built quadcopter using the PX4 flight stack is utilized for the experiments. The drone frame is built using a XILO Phreakstyle Freestyle frame kit, the flight control unit is a  Holybro Kakute H7 that is connected to a ground control station through WIFI. The position and orientation is captured through a motion capture system (OptiTrack) whereas angular velocity and acceleration are measured from an onboard inertial measurement unit (IMU). Both systems have their data fused in a Kalman filter for accurate state estimation. The model parameters for this setup are $l=0.107642$ m, $I_{xx}=0.002261$ kg m$^2$, $I_{yy}=0.002824$ kg m$^2$, $I_{zz}=0.002097$ kg m$^2$, $k_t=0.01$, $g=9.81$ m/s$^2$, $m=0.579902$ kg, $k_{p_{11}}=-5.25$, $k_{p_{12}}=-5.25$, $k_{p_{13}}=3$, $k_{p_{21}}=3.5$, $k_{p_{22}}=3.5$, $k_{p_{23}}=0.35$, $k_{d_{1}}=0.4$, $k_{d_{2}}=0.4$, and $k_{d_{3}}=0.1$.

To demonstrate the applicability of the developed framework to typical quadcopter deployment scenarios where the autopilot is proprietary and unknown, this experiment utilizes the default PX4 autopilot, which is different from the autopilot in \eqref{Equation: Autopilot}. While the PX4 autopilot is able to track the velocity inputs sent by the surrogate pilot, the performance of the real quadcopter employing the PX4 autopilot is substantially different from the performance of a simulated quadcopter employing the autopilot in \eqref{Equation: Autopilot}.

To ensure that the closed-loop model presented in Section \ref{Section: Pilot Model IRL} fits the closed loop model of real quadcopter, the proportional and derivative gains in \eqref{Equation: Autopilot} are manually adjusted so that the response of the model in Section \ref{Section: Pilot Model IRL} employing the autopilot in \eqref{Equation: Autopilot} and the real quadcopter, employing the PX4 autopilot, to velocity commands sent by the surrogate LQR pilot is as close as possible.
 
\subsection{Controller Implementation}
The quadcopter is controlled via an off-board ground control station that implements the surrogate LQR pilot. The objective of the pilot is to return the quadrotor to the origin starting from a given known initial condition using velocity and yaw rate commands. The surrogate pilot implements the control policy that optimizes the cost functional in \eqref{Equation: Pilot Cost Fun}, assuming the linear closed-loop quadrotor model given in \eqref{Equation: Quad Linear_Sys}, with\footnote{The notation $\diag(v)$ represents a diagonal matrix with the elements of the vector $v$ along the diagonal.}
\begin{align}
    Q &=\diag([9.57, 6.91, 2.84, 0, 0, 0, 0, 0, 11.68, 0, 0, 0]) \text{ and }\nonumber\\
    R &= \diag([9.57, 3.48, 14.40, 0.17]).
\end{align}
The pairs $(A,B)$ and $(A, \sqrt{Q})$ are confirmed to satisfy the stabilizability and detectability conditions in Assumption \ref{Assumption: Stab_Det} using  PBH tests in Theorems 14.3 and 16.6 in \cite{ SCC.Hespanha2009}, respectively.

The cost functional is designed under the assumption that the surrogate LQR pilot only penalizes the state variables corresponding to the translational position and the heading. To reduce the number of unknown parameters, the sparsity structure of $Q$ and $R$ is assumed to be known and only the nonzero elements of $Q$ and $R$ are estimated. As a result, the number of unknown parameters in $Q$ is reduced from $78$ to $4$ and the number of unknown parameters in $R$ is reduced from $9$ to $3$, resulting in a total of $85$ unknown parameters. 

To satisfy the FI condition in Definition \ref{Definition: FI IRL}, the ground control station adds an excitation signal onto the velocity commands generated by the surrogate pilot before they are sent to the autopilot. As a result, the final commanded velocity is
\begin{equation}
    U_{cmd} = U_{exc} + U,
\end{equation}
where $U=-K_{EP}X$ is the command generated by the surrogate pilot and $U_{exc}$ is the excitation signal. The signal $U$ is recorded in the history stack.
\subsection{Methods}
A total of $13$ repeated trials are performed to gauge the performance of the developed IRL technique.  In each of the $13$ experiments, the quadcopter is started at a randomly generated hover point contained within the operating area. The surrogate LQR pilot then commands the quadcopter to fly to the origin with a $z$-offset equal to the desired flight height. To ensure that the measured costs are representative of the infinite horizon cost, the controller is run for a time horizon of $200$ s, which is more than 4 times the observed time constant of the surrogate LQR controller. The excitation signal $U_{exc}$ is composed of $4$ sets of $75$ sinusoidal signals. Each set spans a frequency range from $0.001$ Hz to $10$ Hz, with a varying frequency and a magnitude of $0.03$. 

Since the regressor $\Sigma$ is a nonlinear function of the states, relationships between persistence of excitation, number of frequencies in the excitation signal, and number of unknown parameters, well-established in linear systems theory, do not apply to this problem. Using the sufficient conditions developed for linear regressors as a heuristic guideline, the number of frequencies in the excitation signal is initially selected to be roughly equal to the number of unknown parameters, and tuned using trial and error. The magnitude of the excitation signal is also selected using trial and error in simulation. A larger magnitude excitation signal typically results in a smaller condition number of $\Sigma^{\top} \Sigma + \epsilon I$. However, larger excitation magnitudes result in longer quadcopter trajectories, which require a larger flight arena. The excitation signal selected above was tuned using a quadcopter simulator to ensure sufficiently small condition number of $\Sigma^{\top} \Sigma + \epsilon I$ while keeping the quadcopter confined within the flight arena available in the laboratory.

The RHSO is implemented with regularization parameter $\epsilon=0.002$, and data are collected at a sampling rate of $0.08$ seconds using the condition number minimization algorithm described in Section \ref{Section: Pilot Model IRL}. %In addition to the condition number minimization, the history stack is periodically emptied refilled. To implement the history stack refresh, two history stacks are maintained, a main history stack and an auxiliary history stack. Both history stacks are initialized with zero matrices. New data are always added to the auxiliary history stack using the condition number minimization algorithm. If the condition number of $\Sigma^{\top}\Sigma+\epsilon I$, computed using matrices in the the auxiliary history stack, is less than $ 1\times10^9 $ or if 9 seconds have passed since the last refresh, the main history stack is replaced with the auxiliary history stack and the auxiliary history stack is emptied. 
The initial guesses for the unknown weights are randomly generated to be normally distributed in the interval $[-5, 5]$.
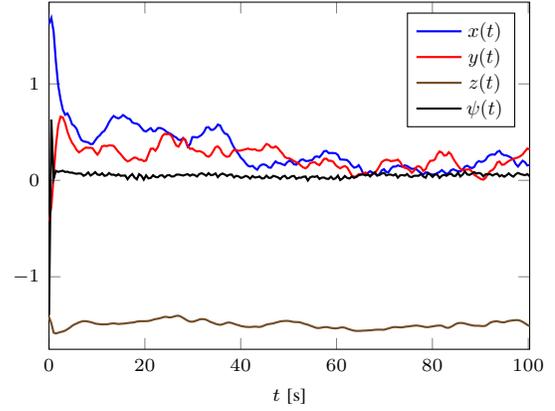
\begin{figure}[t]
    \centering
    \begin{tikzpicture}
    \begin{axis}[
        xlabel={$t$ [s]},
        label style={font=\scriptsize},
        tick label style={font=\scriptsize},
        legend pos = north east,
        legend style={nodes={scale=0.75, transform shape}},
        enlarge y limits=0.05,
        enlarge x limits=0,
        width=0.9\linewidth,
        height=0.7\linewidth,
    ]
        \pgfplotsinvokeforeach{1,...,4}{
            \addplot+ [
                thick,
                mark=none,
                each nth point=20,
                skip coords between index={4180}{8306}
            ] table [x index=0, y index=#1] {data/Position_traj.dat};
        }
        \legend{$x(t)$, $y(t)$, $z(t)$, $\psi(t)$}
    \end{axis}
\end{tikzpicture}
    \caption{Position and heading of the quadcopter in one of the 13 experiments.}
    \label{fig:POS}
\end{figure}
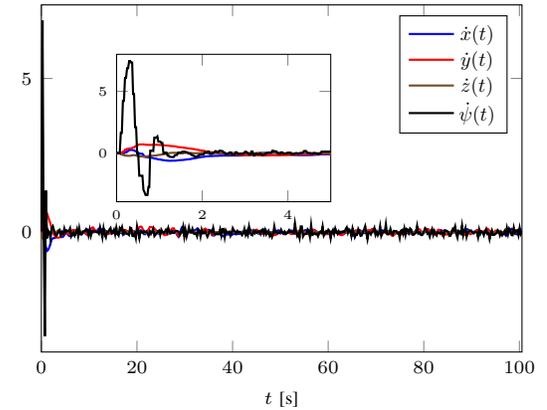
\begin{figure}
    \centering
    \begin{tikzpicture}
    \begin{axis}[
        name=ax1,
        xlabel={$t$ [s]},
        label style={font=\scriptsize},
        tick label style={font=\scriptsize},
        legend pos = north east,
        legend style={nodes={scale=0.75, transform shape}},
        enlarge y limits=0.05,
        enlarge x limits=0,
        width=0.9\linewidth,
        height=0.7\linewidth,
    ]
        \pgfplotsinvokeforeach{1,...,4}{
            \addplot+ [
                thick,
                mark=none,
                each nth point=10,
                skip coords between index={4180}{8306}
            ] table [x index=0, y index=#1] {data/Velocity_Traj.dat};
        }
        \legend{$\dot x(t)$, $\dot y(t)$, $\dot z(t)$, $\dot \psi(t)$}
    \end{axis}

    \begin{axis}[
        name=ax2,
        enlarge y limits=0.05,
        enlarge x limits=0,
        at={($(ax1.south west)+(1cm,2cm)$)},
        width=0.5\linewidth,
        height=0.4\linewidth,
        xmin=0,
        xmax=5,
        tick label style={font=\tiny},
    ]
        \pgfplotsinvokeforeach{1,...,4}{
            \addplot+ [skip coords between index={201}{8306}, thick, mark=none] table [x index=0, y index=#1] {data/Velocity_Traj.dat};
        }
    \end{axis}
\end{tikzpicture}
    \caption{Linear velocity and yaw rate of the quadcopter in one of the 13 experiments.}
    \label{fig:VEL}
\end{figure}
\subsection{Results and Discussion}
The experimental results obtained from one of the 13 flight tests are shown in Figs. \ref{fig:POS}-\ref{fig:QR}. The position of the quadcopter as a function of time is shown in Fig. \ref{fig:POS}, and the linear velocity of the quadcopter as a function of time is shown in Fig. \ref{fig:VEL}. The quadcopter holds position at the origin with a $z$-offset of $1.5$ m and the velocity appears noisy due to the excitation signal. The convergence of $\Delta$ to zero (Fig. \ref{fig:DELTA})\footnote{The notation $\left\Vert\cdot\right\Vert$ is used to denote the euclidean norm when applied to a vector and the Frobenius norm when applied to a matrix.}, combined with the convergence of $\hat{K}_{p}$ to $K_{EP}$ (Fig. \ref{fig:KP}) to zero indicates that the developed technique is able to obtain an equivalent solution (per Definition \ref{Definition: Equivalent Solution Pilot}) to the IRL problem. The experimental results are thus consistent with Theorem \ref{Theorem: Delta Convergence IRL}.

Figs. \ref{fig:DELTA} and \ref{fig:KP} demonstrate that while the feedback policy of the surrogate LQR pilot is estimated correctly, the estimated cost functional is substantially different from the cost functional of the surrogate LQR pilot. This behavior is expected because the underlying IRL problem has multiple equivalent solutions. As indicated by Fig. \ref{fig:ALL_QR}, the cost functional recovered from the data in each of the $13$ experiments converges to different equivalent solutions. The particular equivalent solution recovered in each run depends on the initial guess of the unknown weights used in that run.      
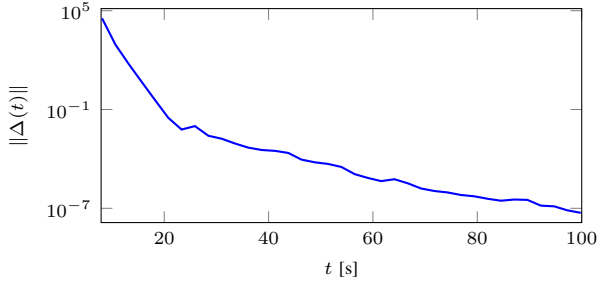
\begin{figure}[t]
    \centering
    \begin{tikzpicture}
    \begin{semilogyaxis}[
        xlabel={$t$ [s]},
        ylabel={$\left\Vert\Delta(t)\right\Vert$},
        label style={font=\scriptsize},
        tick label style={font=\scriptsize},
        legend pos = north east,
        legend style={nodes={scale=0.75, transform shape}},
        enlarge y limits=0.05,
        enlarge x limits=0.002,
        width=0.9\linewidth,
        height=0.5\linewidth,
    ]
        \addplot+ [
            thick,
            mark=none,
            each nth point=5,
            skip coords between index={185}{377}
        ] table [x index=0, y index=1] {data/Delta.dat};
    \end{semilogyaxis}
\end{tikzpicture}
    \caption{A logscale plot of the norm of $\Delta$ as a function of time in one of the $13$ experiments.}
    \label{fig:DELTA}
\end{figure}
\begin{figure}
    \centering
    \begin{tikzpicture}
    \begin{semilogyaxis}[
        xlabel={$t$ [s]},
        ylabel={$\left\Vert\hat{K}_p(t) - K_{EP}\right\Vert$},
        label style={font=\scriptsize},
        tick label style={font=\scriptsize},
        legend pos = north east,
        legend style={nodes={scale=0.75, transform shape}},
        enlarge y limits=0.05,
        enlarge x limits=0.002,
        width=0.9\linewidth,
        height=0.5\linewidth,
    ]
        \addplot+ [
            thick,
            mark=none,
            each nth point=20,
            skip coords between index={4180}{8306}
        ] table [x index=0, y index=1] {data/Kp_Diff.dat};
    \end{semilogyaxis}
\end{tikzpicture}
    \caption{A logscale plot of the induced $2-$norm of the error between the estimated feedback gain and the surrogate pilot's feedback gain as a function of time in one of the $13$ experiments.}
    \label{fig:KP}
\end{figure}
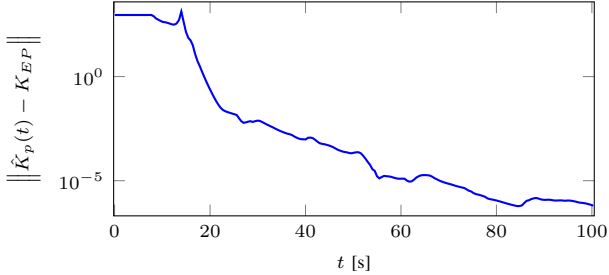
\begin{table}
    \begin{center}
        \begin{tabular}{|c||c|c|}
        \hline
          & RHSO & HSO \\ 
         \hline
         \hline
         Mean$\left(\Vert K_{EP}-\hat{K}_{p} \Vert\right)$ & 2.6997e-08 & NaN  \\
         \hline
         Cov$\left(\Vert K_{EP}-\hat{K}_{p} \Vert\right)$ & 8.3316e-15 & NaN  \\
         \hline
        \end{tabular}
    \end{center}  
    \caption{The RHSO and the HSO \cite{SCC.Self.Coleman.ea2021a} are evaluated by comparing the mean and covariance of the induced $2-$ norm of $\hat{K}_{p} - K_{EP}$ for the 13 tests.}\label{Table: Results}
\end{table}   
From the 13 experiments, it is evident that RHSO finds equivalent solutions for the pilot modeling problem. A sufficiently excited system state is needed to meet the data sufficiency conditions in Definition \ref{Definition: FI IRL}. In this effort, to achieve excitation, an excitation signal is added to the surrogate LQR pilot's command. The excitation signal is designed using trial and error. It is observed in Table \ref{Table: Results} that the convergence of the estimated solution to an equivalent solution is much faster in this quadcopter pilot modeling application than the simulation results shown in \cite{town2022nonuniqueness}. We postulate that the faster convergence can be attributed to the added excitation signal increasing the information content of the data. Furthermore, as evidenced by Table \ref{Table: Results}, the original history stack observer (HSO) in \cite{SCC.Self.Coleman.ea2021a} diverges in this experiment due to nonuniqueness of solutions of the underlying IRL problem. In contrast, the RHSO converges to an equivalent solution.

Selection of the interval used to add data to the history stacks involves important trade-offs. Longer intervals allow larger changes in two subsequent recorded data points, resulting in a lower condition number of $\Sigma^{\top}\Sigma+\epsilon I$; whereas, shorter intervals allow for faster population of the history stacks, which results in better utilization of excitation naturally present during transient response, especially for problems where addition of an excitation signal is not feasible. The tuning of the RHSO also requires selection of an $\epsilon$ to ensure invertibility of $\Sigma^{\top}\Sigma+\epsilon I$. Large values of $\epsilon$ were observed to slow down the convergence rate, a phenomenon for which the authors presently lack an explanation.
        
\section{Conclusion}\label{Section: Conclusion}
The experimental results demonstrate the ability of the RHSO to consistently learn an equivalent solution for a surrogate LQR pilot's cost functional. The estimated cost functional reproduces the surrogate pilot's feedback matrix. The robustness of the algorithm is demonstrated through convergence obtained using randomly generated setpoints and initial guesses for unknown weights.
\begin{figure}[t]
    %\vspace{-0.3in}
    \centering
    \begin{tikzpicture}
    \begin{axis}[
        xlabel={$t$ [s]},
        label style={font=\scriptsize},
        tick label style={font=\scriptsize},
        legend pos = north east,
        legend style={nodes={scale=0.75, transform shape}},
        enlarge y limits=0.05,
        enlarge x limits=0,
        width=0.9\linewidth,
        height=0.5\linewidth,
    ]
        \addplot+ [
            thick,
            mark=none,
            each nth point=20,
            skip coords between index={4180}{8306}
        ] table [x index=0, y index=1] {data/QR_Diff.dat};
        \addplot+ [
            thick,
            mark=none,
            each nth point=20,
            skip coords between index={4153}{8306}
        ] table [x index=0, y index=2] {data/QR_Diff.dat};
        \legend{$\left\Vert\hat{Q}(t) - Q\right\Vert$, $\left\Vert\hat{R}(t) - R\right\Vert$}
    \end{axis}
\end{tikzpicture}
    \caption{A plot of the induced $2-$norm of the error between $\hat Q$ (blue) and $Q$ and $\hat R$ (red) and $R$ as a function of time in one of the 13 experiments.}
    \label{fig:QR}
\end{figure}
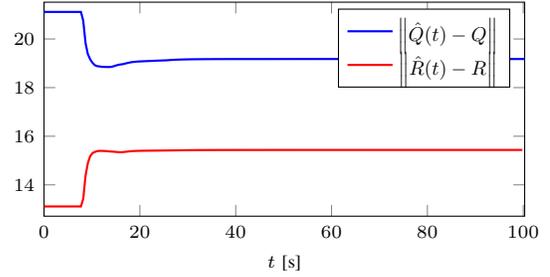

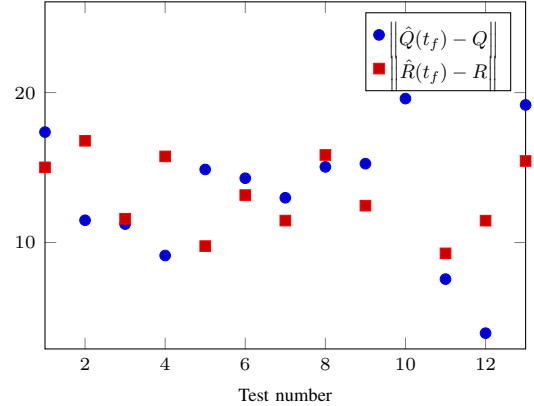
\begin{figure}
    \centering
    \begin{tikzpicture}
    \begin{axis}[
        xlabel={Test number},
        label style={font=\scriptsize},
        tick label style={font=\scriptsize},
        legend pos = north east,
        legend style={nodes={scale=0.75, transform shape}},
        enlarge y limits=0.05,
        enlarge x limits=0,
        width=0.9\linewidth,
        height=0.7\linewidth,
        ymax=25,
    ]
        \addplot+ [only marks] table [x index=0, y index=1] {data/ALL_Q_R.dat};
        \addplot+ [only marks] table [x index=0, y index=2] {data/ALL_Q_R.dat};
        \legend{$\left\Vert\hat{Q}(t_f) - Q\right\Vert$, $\left\Vert\hat{R}(t_f) - R\right\Vert$}
    \end{axis}
\end{tikzpicture}
    \caption{Norm of the error between $Q$ and $\hat{Q}$ and $R$ and $\hat{R}$ obtained at $t = 200$s in the 13 experiments.}
    \label{fig:ALL_QR}
\end{figure}

In solving the pilot modeling problem, the pilot is assumed to be an optimal controller that has full state information and transmits velocity commands to the quadcopter. The results of this paper indicate that this assumption is acceptable for the case where the pilot is a surrogate LQR controller. Further experimentation with human pilots will be required to establish the validity of this assumption in a real-world scenario.

The assumption that excitation signals can be designed so that they do not interrupt a human pilot from performing their mission is reasonable but requires careful tuning of the excitation signal so it does not become a nuisance. Validation of the assumption that the human pilot behaves like to a deterministic LQR controller needs further experimentation with human pilots. Future research will focus on experimentation involving human pilots where the developed IRL method will be used to replicate their performance by learning cost functionals equivalent to the ones being minimized by the pilots. Future work will also involve possible extensions of the developed framework to nonlinear systems and probabilistic models of pilot behavior.
        
\bibliographystyle{IEEEtran}
\bibliography{scc,sccmaster,scctemp,extra}

\begin{IEEEbiography}[{\includegraphics[width=1in,height=1.25in,clip,keepaspectratio]{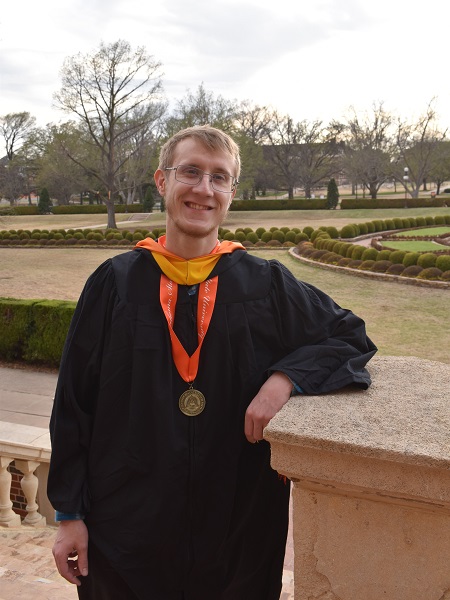}}]{Jared Town} Received his B.S. degree in Mechanical Engineering in 2021 and his M.S. degree in Mechanical and Aerospace Engineering in 2023 from Oklahoma State University. As an undergraduate, he focused on machine design, manufacturing processes, and CAD modeling. For his graduate work, he studied control theory and applications of inverse reinforcement learning. %He has since been employed as a software engineer for FlightSafety.
%%% Don't know if that last line is necessary. Feel free to edit everything after the receiving of degrees, or even that.

\end{IEEEbiography}
\begin{IEEEbiography}[{\includegraphics[width=1in,height=1.25in,clip,keepaspectratio]{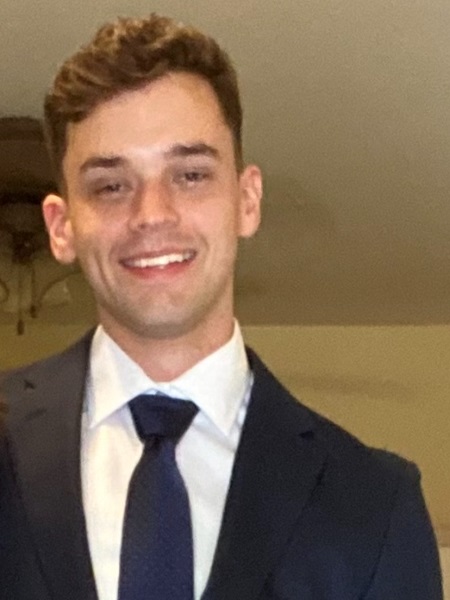}}]{Zachary Morrison}
graduated with a degree in Aerospace Engineering from Oklahoma State University in 2020. He is currently pursing an M.S. degree in Mechanical and Aerospace Engineering at Oklahoma State University.
\end{IEEEbiography}
\begin{IEEEbiography}[{\includegraphics[width=1in,height=1.25in,clip,keepaspectratio]{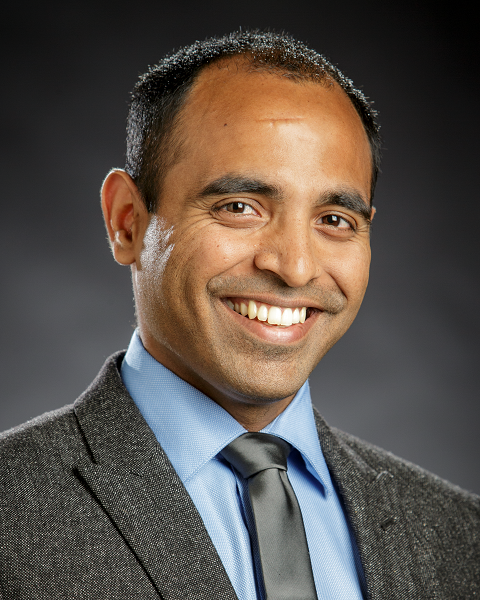}}]{Rushikesh Kamalapurkar}
received M.S. and Ph.D. degrees in 2011 and 2014, respectively, from the Department of Mechanical and Aerospace Engineering at the University of Florida. He is the director of the Systems, Cognition, and Control laboratory at Oklahoma State University. He has published a book, multiple book chapters, over 30 peer reviewed journal papers and over 30 peer reviewed conference papers.
\end{IEEEbiography}

\end{document}